\documentclass{article}
\usepackage[a4paper, total={6in, 8in}]{geometry}
\usepackage{natbib}

\usepackage{tikz}
\usepackage{amsmath}
\usepackage{amsthm}
\usepackage{hyperref}
\hypersetup{
    colorlinks=true,
    linkcolor=blue,
    filecolor=blue,      
    urlcolor=blue,
    citecolor=blue,
    pdfpagemode=FullScreen,
    }
\usepackage{amsfonts}

\usetikzlibrary{shapes.geometric, positioning}

\newcommand{\supp}{\mathsf{supp}}

\title{Find a witness or shatter:\\ the landscape of computable PAC learning.}
\usepackage{times}

\author{Valentino Delle Rose$^1$ \and Alexander Kozachinskiy$^{123}$ \and Cristobal Rojas$^{12}$ \and Tomasz Steifer$^{1234}$}
\date{%
    $^1$Centro Nacional de Inteligencia Artificial, Chile\\%
    $^2$Instituto de Ingeniería Matemática y Computacional, Universidad Católica de Chile\\
    $^3$Instituto Milenio Fundamentos de los Datos, Chile\\%
    $^4$Institute of Fundamental Technological Research, Polish Academy of Sciences\\%
}

\newtheorem{theorem}{Theorem}

\newtheorem{proposition}{Proposition}
\newtheorem{remark}{Remark}
\newtheorem{definition}{Definition}

\begin{document}

\maketitle

\begin{abstract}%
  This paper contributes to the study of CPAC learnability ---a computable version of PAC learning-- by solving three open questions from recent papers. Firstly, we prove that every improperly CPAC learnable class is contained in a properly CPAC learnable class with polynomial sample complexity. This confirms a conjecture by Agarwal et al (COLT 2021). Secondly, we show that there exists a decidable class of hypotheses which is properly CPAC learnable but only with uncomputably growing sample complexity. This solves a question from Sterkenburg (COLT 2022). Finally, we construct a decidable class of finite Littlestone dimension which is not improperly CPAC learnable, strengthening a recent result of Sterkenburg (2022) and answering a question posed by Hasrati and Ben-David (ALT 2023). Together with previous work, our results provide a complete landscape for the learnability problem in the CPAC setting. 
\end{abstract}

\section{Introduction}

The fundamental problem in the theory of Machine Learning is to understand when a given hypothesis class can be learned by an algorithm that has access to finitely many random samples of an unknown objective function. The goal of the learning algorithm, or the \emph{learner} for short, is to select a function that approximates the objective function almost as well as any hypothesis from the given class. The fundamental theorem of statistical learning asserts that such a learner exists if and only if the VC dimension of the class is finite \citep{vapnik1971uniform, blumer1989learnability}. This characterization is concerned with the existence of learners as abstract mathematical functions and does not take into account their computational properties.

Recently, a new framework combining PAC learning and Computability Theory was proposed by~\cite{agarwal2020learnability}. In computable PAC (CPAC) learning, both the learner and the functions  it outputs are required to be \textit{computable}, in the sense that they can be computed by a Turing Machine. As observed in \cite{agarwal2020learnability}, the existence of a computable learner  no longer  follows from a finite VC dimension. Moreover, the computable setting is sensible to aspects of the problem that make no difference in the classical setting. For example, it becomes essential whether the learner is required to be \textit{proper} (i.e,  constrained to only output functions that belong to the hypothesis class) or allowed to be \textit{improper} (can output arbitrary functions). Another issue is whether the \textit{sample complexity}, i.e.,  the number of samples a learner needs in order to work as requested, can always be bounded by a computable function (a setting referred to as \textit{strong} CPAC learning). This raises a number of natural questions ---which of these aspects of the problem actually lead to different versions of computable learnability?

 Significant progress was made by \cite{sterkenburg2022characterizations}, who gave a characterization of proper strong CPAC learning in terms of the computability of a \textit{Empirical Risk Minimizer} (ERM) and who constructed a class of finite VC dimension which is not CPAC learnable, even in the improper sense. Very recently, \cite{hasrati2023} gave the computability-theoretic perspective on a related framework of \textit{online learning}, further improving our understanding of the learning problem in the computable setting. 


\paragraph{\textbf{Main results.}} 
The current paper contributes to this line of research, in particular, by solving three open problems raised in these recent papers.
First, we provide a characterization of CPAC learning in the \textit{improper} setting, i.e., when the learner is allowed to output a function outside of the given hypothesis  class. For that, we introduce the \textit{effective VC dimension}. The classical VC-dimension of a hypothesis class $\mathcal{H}$ can be defined as the minimal $k$ such that for any tuple of $k+1$ distinct natural  numbers one can indicate a Boolean function on them which is not realizable by  hypotheses of $\mathcal{H}$. In the effective version of VC dimension,  there also must be an \emph{algorithm}, transforming a tuple of $k+1$ distinct natural numbers into a Boolean function, not realizable on them by $\mathcal{H}$. Our first result states that a hypothesis class is improperly CPAC learnable if and only if  its effective VC dimension is finite. As a byproduct, we obtain that  every improperly CPAC learnable class is in fact a subclass of a properly CPAC learnable class, settling a conjecture formulated by~\cite{agarwal2021open}. Secondly, we show that there exists a decidable class of hypotheses $\mathcal{H}$ that has proper computable learners, but only those whose sample complexity cannot be bounded from above by a computable function. This separates CPAC learning from strong CPAC learning in the \textit{proper} setting, solving a question asked by \cite{sterkenburg2022characterizations}. Finally, we strengthen a theorem of \cite{sterkenburg2022characterizations}, who constructed a decidable class of finite VC dimension which is not improperly CPAC learnable. We show that such a class can be constructed to even have a finite Littlestone dimension, providing an answer to a question posed by \cite{hasrati2023} in the context of online learning.
Altogether, our results provide a comprehensive landscape for the learnability problem in the computable setting.

\color{black}

\paragraph{\textbf{Organization of the paper.}}In Section \ref{sec:prel}, we briefly recall the  classical PAC learning framework. In Section \ref{sec:CPAC}, we provide a detailed overview of  computable PAC learning, go through results and open problems from previous works, and then present precise statements of our results. Proofs are given in the subsequent sections.

\section{Preliminaries}
\label{sec:prel}
\subsection{Notation} For any two sets $A$ and $B$, we denote by $B^A$ the set of functions $f\colon A\to B$. If $f\colon A \to \{0, 1\}$, by the support of $f$, denoted by $\supp(f)$, we mean $f^{-1}(1)$. 

\subsection{Classical PAC learning}

In this section, we briefly introduce the classical PAC learning framework. We only work over the domain $\mathbb{N}$. Thus, a \emph{hypothesis class} $\mathcal{H}$ is an arbitrary subset of  $\{0,1\}^\mathbb{N}$. Elements of $\mathcal{H}$ will be called \emph{hypotheses}.  We will  say that a class is \emph{finitely supported} if it consists of only hypotheses with finite support. A \emph{sample} of size $n$ is an element of  $(\mathbb{N}\times\{0, 1\})^n$. A \emph{learner} is any function $A$ from the set of all samples (that is, from $\bigcup_{n\in\mathbb{N}}(\mathbb{N}\times\{0, 1\})^n$) to $\{0, 1\}^\mathbb{N}$. Now, if $D$ is a probability distribution over $\mathbb{N}\times \{0, 1\}$, then the \emph{generalization error} of $h\in \{0,1\}^\mathbb{N}$ with respect to $D$ is
\[L_D(h) = \Pr_{(x, y)\sim D}[h(x)\neq y].\]
When we write $S\sim D^m$, we mean that the coordinates of the sample $S = ((x_1, y_1), \ldots, (x_m, y_m))$ were drawn independently $m$ times from $D$. In the PAC learning framework, the learner receives a sample $S\sim D^m$ for some sufficiently large $m$. The learner's task is to select, with high probability over $D^m$, some function $f\colon\mathbb{N}\to\{0,1\}$ whose generalization error with respect to $D$ is close to the best possible generalization error achievable by functions in some known, a priori given, hypothesis class $\mathcal{H}$. 

\begin{definition}\label{def:PAC}
Let $\mathcal{H}$ be a hypothesis class and $A$ be a learner. We say that $A$ \textbf{PAC learns} $\mathcal{H}$ if for every $n\in\mathbb{N}$ there exists $m_n$ such that for every $m\ge m_n$ and for every probability distribution $D$  over $\mathbb{N}\times \{0, 1\}$ we have:
\begin{equation}
\label{defi}
    \Pr_{S\sim D^m}[L_D(A(S)) \le \inf_{h\in\mathcal{H}} L_D(h) +1/n]\ge 1 - 1/n.
    \end{equation}
If there exists $A$ which PAC learns $\mathcal{H}$, then $\mathcal{H}$ is called \textbf{PAC learnable}.
    
If $A$ PAC learns $\mathcal{H}$, then the  \textbf{sample complexity} of $A$ w.r.t.~$\mathcal{H}$ is a function $m\colon \mathbb{N}\to\mathbb{N}$, where $m(n)$ is the minimal natural number $m_n$ such that, for all $m\ge m_n$, \eqref{defi} holds for $n$.

We say that $A$ is \textbf{proper} for $\mathcal{H}$ if it only outputs hypotheses from $\mathcal{H}$. Otherwise, we say that $A$ is \textbf{improper} for $\mathcal{H}$.
 \end{definition}

A classical result of learning theory is that PAC learnability admits a combinatorial characterization via the so-called (Vapnik-Chervonenkis) VC dimension. More specifically,
let $x_1, \ldots, x_k$ be $k$ distinct natural numbers.
We say that a hypothesis class $\mathcal{H}$
\emph{shatters} $x_1, \ldots, x_k$ if for all $2^k$ functions $f\colon\{x_1, \ldots, x_k\}\to\{0, 1\}$ there exists $h\in\mathcal{H}$ with $f(x_1) = h(x_1), \ldots, f(x_k) = h(x_k)$. The  \emph{VC dimension} of $\mathcal{H}$ is the maximal natural $k$ for which there exist $k$ distinct natural numbers that are shattered by $\mathcal{H}$. If such $k$ distinct natural numbers exist for all $k$, the VC dimension of $\mathcal{H}$ is $+\infty$. It turns out that a class $\mathcal{H}$ is PAC learnable if and only if its VC dimension is finite.

Another classical result is that any PAC learnable class $\mathcal{H}$ can be learned by a specific type of learners called \emph{Empiricial Risk Minimizers} (or ERM for short). To define them, we first have to define the \emph{error of a hypothesis on a sample}. Namely, let $S = ((x_1, y_1), \ldots, (x_m, y_m))$ be a sample and $h\colon\mathbb{N}\to\{0,1\}$. The error of $h$ on $S$ is defined as:
\[L_S(h) = \frac{\left|\{i\in \{1, \ldots, m\} \mid h(x_i) \neq y_i\}\right|}{m}.\]

A learner $A$ is an \emph{ERM} for a class $\mathcal{H}$ if for every sample $S$ we have $A(S) \in \arg\min_{h\in\mathcal{H}} L_S(h)$. In other words, for a given sample $S$ an ERM outputs a hypothesis from $\mathcal{H}$ with the least error on $S$. Note that ERM might be not unique (there might be more than one hypothesis attaining the minimum of $L_S(h)$). Finally, there is another relevant property of a class $\mathcal{H}$ that we will also use. We say that a class of hypothesis $\mathcal{H}$ has the \emph{uniform convergence property} if for every $n$ there exists $m_n$ such that for every $m \ge m_n$ and for every probability distribution $D$, 
    \begin{equation}
    \label{UCP}
    \Pr_{S\sim D^m} \left[ \forall h \in H, \ |L_D(h) - L_S(h)| \le \frac{1}{n} \right] \ge 1 - 1/n.
    \end{equation}

The fundamental theorem of statistical learning links together VC dimension, learnability, Empirical Risk Minimization, and the uniform convergence property.

\begin{theorem}[\citet{vapnik1971uniform}, \cite{blumer1989learnability}]\footnote{The definition of VC dimension and its connection with the uniform convergence property were established in \cite{vapnik1971uniform}, while the relationship with PAC learnability is due to \cite{blumer1989learnability}. See e.g.~\cite{MLbook2014}, Theorem 6.7 for an explanatory presentation.}\label{thm:fundamental}

For any class of hypotheses $\mathcal{H}$, the following conditions are equivalent:
\begin{itemize}
 \item $\mathcal{H}$ is PAC learnable,
\item the VC dimension of $\mathcal{H}$ is finite;
    \item every ERM for $\mathcal{H}$ PAC learns $\mathcal{H}$, and its sample complexity w.r.t.~$\mathcal{H}$ is bounded by $O(n^2\log n)$, where the constant hidden in $O(\cdot)$ depends only on $\mathcal{H}$;
    \item $\mathcal{H}$ has the uniform convergence property.
\end{itemize}
\end{theorem}

Similarly to VC dimension for PAC learnability, there is a combinatorial property that characterizes \emph{online learnability} (see \cite{littlestone1988} and \cite{bendavid2009}) and private PAC learning \citep{alon2022private}. Consider a full rooted binary tree $T$ of depth $d$ where   each non-leaf node is labeled by some $x \in \mathbb{N}$.
We say that $T$ is \emph{shattered} by a hypothesis class $\mathcal{H}$ if for every leaf $l$ of $T$ there exists a hypothesis $h\in\mathcal{H}$ which \emph{leads to $l$ in $T$}. I.e., $l$ can be obtained by descending from the root of $T$ in the following manner: if we are in some non-leaf node $v$ (in the beginning, $v$ is the root), and if $x\in\mathbb{N}$ is the label of $v$ in $T$, then we go to the left child of $v$ if $h(x) = 0$   and we go to the right child of $v$ if $h(x) = 1$, and this continues until we reach a leaf.  The \emph{Littlestone dimension of} $\mathcal{H}$ (denoted by $\mathrm{Ldim(\mathcal{H})}$) is defined as the maximal depth of a tree which is shattered by $\mathcal{H}$. If $\mathcal{H}$ shatters trees of arbitrary depth, then we set $\mathrm{Ldim(\mathcal{H})}= \infty$. It is not hard to see that for every hypotheses class $\mathcal{H}$, the VC dimension of $\mathcal{H}$ is at most $\mathrm{Ldim(H)}$ (see e.g.~\cite{MLbook2014}, Theorem 21.9).

\section{Computable PAC learning: previous results and our contribution}
\label{sec:CPAC}

In the computable version of PAC learning (introduced by~\cite{agarwal2020learnability}) one considers only \emph{computable} learners, i.e. learners that can be implemented by an algorithm in the strict sense of the word (a Turing machine). 
In more detail, we say that a learner $A$ is computable if there exists an algorithm which, given a sample $S$, outputs a description of a computer program (formally, a Turing machine), implementing $h=A(S)$ (the function that $A$ outputs on $S$). In particular, with this program, we can evaluate $h$ on any natural number.  
In this paper, we mainly deal with learners that only output finitely supported hypotheses. To show that such a learner $A$ is computable, it is enough to exhibit an algorithm that transforms $S$ into $\supp(A(S))$ -- since $\supp(A(S))$ is finite, one can easily construct a Turing machine that outputs 1 exactly on $\supp(A(S))$. Requiring a learner to be computable naturally leads to the following definition: 

\begin{definition}
\label{def:cpac}
    A  hypothesis class $\mathcal{H}$ is \textbf{computably PAC learnable} (or \textbf{CPAC learnable} for short) if there exists a computable learner\footnote{We stress that in this definition the learner might be \emph{improper}, in the sense that it is allowed to output functions that do not necessarily belong to $\mathcal{H}$.} that PAC learns it.
 \end{definition}

\begin{remark} Previous authors make an explicit assumption for CPAC learnability that a hypothesis class must contain only  computable functions. In this paper, we drop this assumption because it is never used in the proofs and we consider it unnecessarily restrictive. Nevertheless, all our results hold with this assumption as well. Moreover, all examples of classes in this paper are computationally simple. Namely, they consist  only of functions with finite support. In fact, all these classes are decidable, meaning that there is an algorithm that, when provided the support of  a given finitely supported function $h$,
 decides, whether $h$ belongs to the class.  
 \end{remark}

The question of whether or not CPAC learnability is different from PAC learnability is not obvious and was left as an open problem in~\cite{agarwal2020learnability}. It was recently solved by~\cite{sterkenburg2022characterizations}, who gave an example of a hypothesis class with VC-dimension 1 which is not CPAC learnable. In fact, his class  consists of hypotheses with finite support and is decidable.

Agarwal et al.~also introduced a more restrictive version of CPAC learnability by constraining learners to only output functions from the given hypothesis class $\mathcal{H}$.

\begin{definition}
\label{def:prop-cpac}
    A  hypothesis class $\mathcal{H}$ is \textbf{properly CPAC learnable}  if there exists a computable learner that PAC learns $\mathcal{H}$ and that is proper for $\mathcal{H}$. 
 \end{definition}

 \begin{remark} The reader should be warned that the choice of names for these definitions is not consistent across the literature. Here we have followed ~\cite{agarwal2020learnability}. But for example, \cite{sterkenburg2022characterizations} calls CPAC learnability what we have called here proper CPAC learnability.  
 \end{remark}

In the classical setting, the requirement of being proper does not change anything. Indeed, by Theorem \ref{thm:fundamental}, whenever some learner PAC learns $\mathcal{H}$, we have that any ERM for $\mathcal{H}$
PAC learns it as well, and any ERM is proper by definition. However, in the computable setting, the existence of a computable learner does not necessarily guarantee the computability of some ERM, or of any other proper learner. In fact, as shown by \cite{agarwal2020learnability}, the set of properly CPAC learnable classes is indeed strictly contained in the set of CPAC learnable ones.  Their example consists of a certain decidable class of hypotheses $\mathcal{H}$ made of functions with the support of size 2 (and thus PAC learnable since its VC dimension is at most 2) which does not admit any proper computable learner, but for which an (improper) computable learner can be easily constructed. Even better, in this case, the class $\mathcal{H}$ is included in the decidable class $\mathcal{H}_2$ of all functions with size-2 support, which clearly admits a computable ERM and is therefore properly CPAC learnable. A question then naturally arises: 
\begin{center}\emph{is every CPAC learnable class $\mathcal{H}$ contained in some properly CPAC learnable class $\widehat{\mathcal{H}}$?}
\end{center}
Observe that indeed all subclasses of a properly CPAC learnable $\widehat{\mathcal{H}}$ are (improperly) CPAC learnable (via the same computable learner that PAC learns $\widehat{\mathcal{H}}$). However, it is not clear if this is the only way a class $\mathcal{H}$ can be CPAC learnable. This question was raised by both \cite{sterkenburg2022characterizations} and \cite{agarwal2021open}. In this paper, we shall answer it in the affirmative (see Theorem \ref{imp_char} below).

A similar analysis can be carried out with a focus on a different aspect of the problem --sample complexity. As we have just mentioned, in the classical setting, a class $\mathcal{H}$ that is PAC learnable, is automatically learnable by any ERM, and the sample complexity of ERMs is only polynomial in $n$. Once again, in the computable setting, ERMs are not necessarily computable and therefore, for a given computable learner, even the computability of its sample complexity function is not guaranteed. Motivated by this issue, Sterkenburg introduced the following \emph{strong} variant of CPAC learnability.

\begin{definition}
\label{def:scpac}
    A  hypothesis class $\mathcal{H}$ is \textbf{strongly CPAC learnable} (or SCPAC learnable for short) if there exists a computable learner that PAC learns $\mathcal{H}$ and whose sample complexity can be bounded by some total computable function.  Moreover, if there exists a learner with these properties which is proper for $\mathcal{H}$, then $\mathcal{H}$ is \textbf{properly} SCPAC learnable.
 \end{definition}
 Sterkenburg obtained a characterization of \emph{proper} SCPAC learnability.
 \begin{theorem}[\cite{sterkenburg2022characterizations}, Theorem 2] \label{thm:proper-SCPAC}
    A class of hypotheses $\mathcal{H}$ is properly SCPAC learnable if and only if  it has a computable ERM and its VC dimension is finite. 
\end{theorem}

One interesting consequence of this result is that whenever $\mathcal{H}$ has a proper computable learner whose sample complexity is bounded by some total computable function, no matter how fast it grows, $\mathcal{H}$ automatically has some other computable learner whose sample complexity is just polynomial (because by Theorem \ref{thm:fundamental}, every ERM has polynomial sample complexity). As a main problem, Sterkenburg left open the following natural question:

\begin{center}\emph{is proper CPAC learnability equivalent to proper SCPAC learnability?}
\end{center}
 In this paper, we solve this problem by showing that these two notions can be separated by a decidable class of hypotheses, see Theorem \ref{separation} below.

\subsection{Statements of our results.}

We start with a characterization of CPAC learnability via \emph{effecitve VC-dimension}. 

\begin{definition}
    Let $\mathcal{H}\subseteq\{0, 1\}^\mathbb{N}$ be a hypothesis class.
    \textbf{A $k$-witness of VC dimension for $\mathcal{H}$} is any function $w$, whose domain is the set of all increasing $(k+1)$-tuples of natural numbers and whose range is $\{0, 1\}^{k+1}$, such that for every $x_1 < x_2 < \ldots < x_{k+1}\in\mathbb{N}$ and for every $h\in \mathcal{H}$ we have:
    \[\big(h(x_1), \ldots, h(x_{k+1})\big)\neq w((x_1, \ldots, x_{k +1})).\]

    The \textbf{effective VC-dimension} of $\mathcal{H}$ is the minimal $k\in\mathbb{N}$ such that there exists a \textbf{computable} $k$-witness of VC-dimension for $\mathcal{H}$. If no such $k$ exists, then the effective VC-dimension of $\mathcal{H}$ is infinite.
    
\end{definition}

In other words, for every  $x_1 < x_2 < \ldots < x_{k+1}$, the witness function $w$ provides a Boolean function on $x_1, \ldots, x_{k+1}$ which is not realizable by hypotheses from $\mathcal{H}$. The minimal $k$ for which such $w$ exists (possibly, not computable) is equal to the ordinary VC dimension of $\mathcal{H}$. Now, if we consider only computable witnesses, we obtain the effective VC dimension. Thus, the effective VC dimension can only be larger than the ordinary one.

Sterkenburg proved that being CPAC learnable implies having finite effective VC dimension (and used this to give an example of a class that is PAC learnable but not CPAC learnable):

\begin{proposition}[\cite{sterkenburg2022characterizations}, Lemma 1]\label{k-wit}
    If $\mathcal{H}$ is CPAC learnable, then it admits a computable $k$-witness of VC dimension, for some natural number $k$.
\end{proposition}

We show that these two conditions are, in fact, equivalent.
Moreover, they are equivalent to being contained in a properly SCPAC learnable class. Thus, our next theorem settles a question of \cite{sterkenburg2022characterizations,agarwal2021open}.

\begin{theorem}
\label{imp_char}
For every $\mathcal{H}$, the following 3 conditions are equivalent:
\begin{itemize}
    \item (a) Effective VC dimension of $\mathcal{H}$ is finite;
    \item (b) $\mathcal{H}$ is  CPAC learnable;
    \item (c) $\mathcal{H}$ is contained in some properly SCPAC learnable hypothesis class $\widehat{\mathcal{H}}$.
\end{itemize}
\end{theorem}
\begin{proof}
    Section \ref{sec:imp}.
\end{proof}

\color{black}

We note that the previous result also shows that  CPAC learnability is equivalent to SCPAC learnability, as well as to being contained in some properly CPAC learnable class (as the latter two properties are stronger than condition $(b)$ but weaker than condition $(c)$ of Theorem \ref{imp_char}).

Next, we observe that one can give a characterization of proper CPAC learnability which is similar in spirit to Proposition \ref{thm:proper-SCPAC}, Sterkenburg's characterization of proper SCPAC learnability. For that, we need the following relaxed version of ERMs.
\begin{definition}
    Let $\mathcal{H}\subseteq \{0,1\}^\mathbb{N}$ be an hypothesis class. A learner $A$ is called an \textbf{asymptotic ERM} for $\mathcal{H}$ if it outputs only hypotheses from $\mathcal{H}$ and if there exists  an infinite sequence $\{\varepsilon_m\in [0, 1]\}_{m = 1}^\infty$, converging to $0$ as $m\to\infty$, such that for every sample $S$ we have that:
\[L_S(A(S)) \le \inf_H L_S(h) + \varepsilon_{|S|}.\]
\end{definition}

Just like the existence of a computable ERM characterizes proper SCPAC learnability, proper CPAC learnability boils down to the existence of a computable asymptotic ERM.
\begin{proposition}
\label{cpac_char}
A hypothesis class $\mathcal{H}\subseteq \{0,1\}^\mathbb{N}$ is properly CPAC learnable if and only if its VC dimension is finite and
it has a computable asymptotic ERM.
\end{proposition}
\begin{proof}
    Section \ref{sec:asy}.
\end{proof}

From the proof of Proposition \ref{cpac_char} it is not hard to see that the error $\varepsilon_m$ is to an asymptotic ERM for $\mathcal{H}$ exactly as the sample complexity function $m(n)$ is to a corresponding proper learner. It follows that $\varepsilon_m$ can be bounded above by a computable function that decreases to $0$, exactly when $\mathcal{H}$ is properly SCPAC learnable. We use this observation to answer another open question of Sterkenburg. Namely, that not every properly CPAC learnable class is properly SCPAC learnable.

\begin{theorem}\label{separation}
\label{example}
There is a decidable class of finitely supported hypothesis $\mathcal{H}$ which is properly CPAC learnable but not properly SCPAC learnable.
\end{theorem}
\begin{proof}
    Section \ref{sec:example}.
\end{proof}

Our results, together with previous works, fully determine the landscape of computable PAC learning (see Figure \ref{res}). First, by Theorem \ref{example},
we have that the set of properly SCPAC learnable classes is strictly included in the set of properly CPAC learnable classes. Next, due to the example of~\cite{agarwal2020learnability}, the set of properly CPAC learnable classes is strictly included in the set of CPAC learnable classes. On the other hand, CPAC learnable classes coincide with SCPAC learnable classes, as well as with classes that are subsets of proper CPAC (or SCPAC) learnable superclasses, by Theorem \ref{cpac_char}. Finally, by the construction of~\cite{sterkenburg2022characterizations}, the set of CPAC learnable classes is strictly included in the set of PAC learnable classes. We were able to strengthen this separation by constructing such an example with not only a finite VC dimension but even a finite Littlestone dimension. More precisely, we prove the following result.

\begin{figure}[h]

\begin{center}
	\begin{tikzpicture}
	\draw[rounded corners, fill=black, opacity=0.25] (2, 0) rectangle (8,1) {};
	\node at (5,.5) (a) {\large prop.~SCPAC};
	\draw[rounded corners, fill=black, opacity=0.2] (1, -.5) rectangle (9,2) {};
	\node at (4,1.5) (b) {\large prop.~CPAC};
	\draw[rounded corners, fill=black, opacity=0.15] (-1.5, -1) rectangle (10,3) {};
	\node at (4.5,2.5) (c) {\large CPAC\, = \, SCPAC\, =  $\subseteq$prop.~CPAC\, = $\subseteq$prop.~SCPAC};
    \draw[rounded corners, fill=black, opacity=0.05] (-2, -2) rectangle (12,4) {};
    \node at (6, 3.5) (e) {\large PAC};
    \end{tikzpicture}
 \caption{
The landscape of computable PAC learning. Note that the strict inclusions hold even in the case of decidable classes of hypotheses.}
\label{res}
\end{center}

\end{figure}
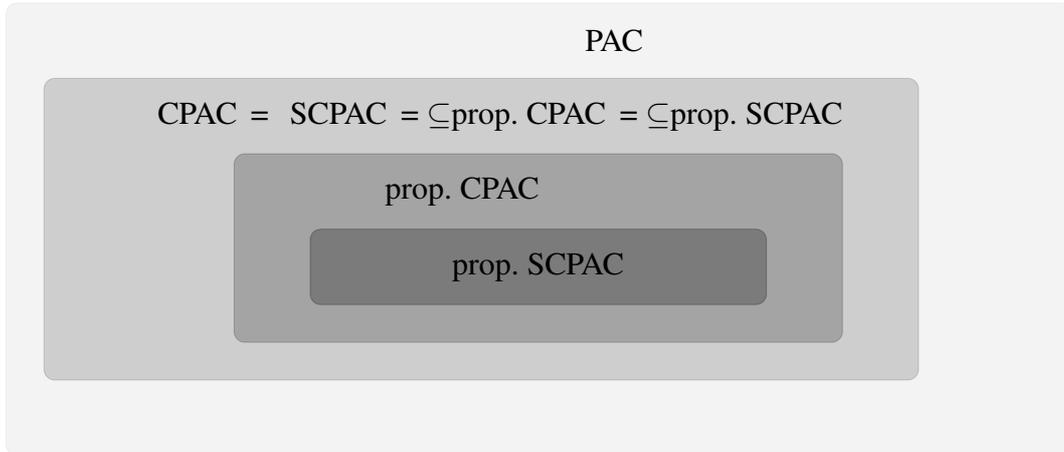

\begin{theorem} \label{thm:finiteLdim}
    There is a decidable class of finitely supported hypotheses $\mathcal{H}$ with $\mathrm{Ldim}(\mathcal{H}) = 1$ which is not (improperly) CPAC learnable.
\end{theorem}
\begin{proof}
    Section \ref{sec:online}.
\end{proof}
This theorem answers an open question from~\cite{hasrati2023}. It also establishes the separation between the classical online learnability (recently shown to be equivalent to private PAC by \cite{alon2022private}) and its computable counterpart,  even for decidable classes of hypotheses (see details in \cite{hasrati2023}, section 6). 

\section{Proof of Theorem \ref{imp_char}}
\label{sec:imp}
 \begin{proof} Implication $(b) \implies (a)$ has been already shown in \cite{sterkenburg2022characterizations} (Proposition \ref{k-wit} above), and $(c) \implies (b)$ follows directly from the definitions. It remains to establish that $(a)\implies (c)$.

Assume that the effective VC dimension of $\mathcal{H}$ is finite. Then for some $k\in\mathbb{N}$ there exists 
 a computable $k$-witness $w$ of VC dimension for $\mathcal{H}$. Let us say that a function $h\colon\mathbb{N}\to\{0,1\}$ is $\emph{good}$ if $h$ has finite support and for every $x_1< x_2 < \ldots <x_{k +1} < \max\supp(h)$ we have that $h$ disagrees with $w$ on $x_1\ldots, x_{k+1}$, i.e.,
    $\big(h(x_1), \ldots, h(x_{k+1})\big)\neq w((x_1, \ldots, x_{k +1}))$.
Let $\mathcal{H}_{good}$ denote the set of all good  $h\colon\mathbb{N}\to\{0,1\}$.
Define $\widehat{\mathcal{H}} = \mathcal{H}\cup \mathcal{H}_{good}$. Obviously, $\mathcal{H}\subseteq \widehat{\mathcal{H}}$.  It remains to show that $\widehat{\mathcal{H}}$ is SCPAC learnable. By Theorem \ref{thm:proper-SCPAC}, it is sufficient to show two things: that $\widehat{\mathcal{H}}$ has finite VC dimension and that $\widehat{\mathcal{H}}$ has computable ERM. 

We first show that the VC dimension of $\widehat{\mathcal{H}}$ is at most $k+1$. Indeed, take any $x_1 < x_2 < \ldots < x_{k +2}\in\mathbb{N}$. We claim that
\begin{equation}
\label{h_eq}
\big(h(x_1), \ldots, h(x_{k+1}), h(x_{k+2})\big) \neq  \big(w((x_1, \ldots, x_{k +1})), 1\big) \text{ for all } h\in\widehat{\mathcal{H}}.
    \end{equation}
Indeed, if $h\in \mathcal{H}$, then $\big(h(x_1), \ldots, h(x_{k+1})\big) \neq w((x_1, \ldots, x_{k +1}))$ because $w$ is a $k$-witness of VC dimension for $\mathcal{H}$, and hence \eqref{h_eq} holds as well. Now, assume that  $h\in \mathcal{H}_{good}$. If $h(x_{k+2}) = 0$, then \eqref{h_eq} holds. In turn, if $h(x_{k+2}) = 1$, observe that $x_1 < x_2 < \ldots < x_{k+1} < x_{k+2} \le \max\supp(h)$. Hence, by definition of a good hypothesis, we  have $\big(h(x_1), \ldots, h(x_{k+1})\big) \neq w((x_1, \ldots, x_{k +1}))$, and this implies \eqref{h_eq}.

It remains to show that $\widehat{\mathcal{H}}$ has a computable ERM. The key is to note that for any given sample $S = ((x_1, y_1), \ldots, (x_m, y_m))$ with $M = \max\{x_1, \ldots, x_m\}$, in order to find a hypothesis with minimal error on $S$, it is enough to go through all good hypotheses whose support is a subset of $\{0, 1, \ldots, M\}$. There are finitely many of them, and we can effectively list them because $w$ is computable. It remains to show that the resulting learner is an ERM for $\mathcal{H}$. For that, it is enough to establish that for any $h\in\widehat{\mathcal{H}}$ there exists a good $h_1$ with $\supp(h_1) \subseteq \{0, 1,\ldots, M\}$ such that $L_S(h) = L_S(h_1)$. Indeed, consider $h_1$ that coincides with $h$ on $\{0,1,\dots,M\}$ and equals $0$ otherwise (so that $\supp(h_1)\subseteq \{0, \ldots, M\}$).  Clearly, $L_S(h) = L_S(h_1)$ because $S$ only involves numbers up to $M$. It remains to show that $h_1$ is good, i.e., that it disagrees with $w$ on every tuple $x_1 < x_2 < \ldots < x_{k+1}$ with $x_{k+1} < \max\supp(h_1)$. Indeed, note that $\max\supp(h_1)\le M$. Hence, $h$ agrees with $h_1$ on  $x_1 < x_2 < \ldots < x_{k+1}$ and on $x_{k+2} = \max\supp(h_1)$ (which means that $h(x_{k+2}) = 1$). Thus, $h_1$ must disagree with $w$ on $x_1 < x_2 < \ldots < x_{k+1}$ because otherwise
we have $(h(x_1), \ldots, h(x_{k+1})) = w((x_1, \ldots, x_{k+1}))$ and $h(x_{k+2}) = 1$, which
contradicts \eqref{h_eq} for $h$.

\end{proof}

\section{Proof of Proposition \ref{cpac_char}}
\label{sec:asy}
\begin{proof}
    First, assume that the VC dimension of $\mathcal{H}$ is finite and that it admits a computable asymptotic ERM $A$. We show that $A$ PAC learns $\mathcal{H}$. Since $A$ is computable and only outputs functions from $\mathcal{H}$ by definition, this means that $\mathcal{H}$ is properly CPAC learnable.  Since the VC dimension of $\mathcal{H}$ is finite, Theorem \ref{thm:fundamental} ensures that it has the uniform convergence property. Let then $m_n$ be such that for every $m\ge m_n$ and every probability distribution $D$, it holds that
    \[
    \Pr_{S\sim D^{m}} \left[ \forall h \in \mathcal{H}, \ |L_D(h) - L_S(h)| \le \frac{1}{3n} \right] \ge 1 - 1/n.
    \]
    Since $A$ only outputs hypothesis from $\mathcal{H}$, we have that with probability at least $1-1/n$, it holds that $|L_D(A(S)) - L_S(A(S))| \le \frac{1}{3n}$, as well as $|L_D(h^*) - L_S(h^*)| \le \frac{1}{3n}$ for any $h^*\in\arg\min_{h \in \mathcal{H}} L_S(h)$. 
     Furthermore, by definition of asymptotic ERM, there exists a sample size $m_n'$ such that for any sample of size at least $m_n'$ we have 
    \[
    L_S(A(S)) \le L_S(h^*) + \frac{1}{3n}.
    \]
    It follows that for any sample $S$ of size $m\ge\max(m^1_n, m^2_n)$,
    \[
    \Pr_{S\sim D^m}\left[L_D(A(S)) \le \inf_{h\in\mathcal{H}} L_D(h) + \frac{1}{n}\right]\ge 1 - 1/n
    \]
    holds, which shows that $A$ is a PAC learner for  $\mathcal{H}$.

    Now assume that $\mathcal{H}$ is CPAC learnable and let $A$ be a computable learner. We construct a computable asymptotic ERM $\widehat{A}$. For any sample $S$, consider the probability distribution $D(S)$ which assigns probability $\frac{1}{|S|}$ to any element $(x,y) \in S$ (in case $S$ has repetitions, we simply sum the corresponding probabilities). Note that with this definition for $D(S)$ we have that 
      \begin{equation}\label{sampled}
    L_S(h) = L_{D(S)}(h) \quad \text{ for all } h\in \mathcal{H}. 
     \end{equation}
    Let us denote by $\widehat{\mathcal{S}}$ the set of all samples $S'$ of size $m=|S|$ that can be drawn from $D(S)$. We then define the output of $\widehat{A}$ on $S$ by 
   \[
   \widehat{A}(S) = \arg\min_{\{A(S')\,:\,S'\in\widehat{S}\}} L_S(A(S')).
   \]

    It is clear that $\widehat{A}$ is computable since $A$ and $D(S)$ are. 
   To show that $\widehat{A}$ is an asymptotic ERM, let $m(n)$ be the sample complexity of $A$ w.r.t.~$\mathcal{H}$. Then, for every $n$ and every $m$ such that $m(n) \le m < m(n+1)$, let $\epsilon_m = \frac{1}{n}$. Since $A$ is a PAC learner for $\mathcal{H}$, we have that for such an $m$ and $S'\sim D(S)^{m}$, with probability at least $1-1/n$ it holds that
    \[
    L_{D(S)}(A(S'))\leq \inf_{h\in\mathcal{H}} L_{D(S)}(h) + \epsilon_{|S|}.
    \]
    This means that there exists some $S^*\in \widehat{S}$ for which this holds, which together with equality (\ref{sampled}) above gives $L_S(\widehat{A}(S))
    \leq \inf_h L_S(h) + \epsilon_{|S|}$, as it was to be shown.
    \end{proof}

\section{Proof of Theorem \ref{example}}
\label{sec:example}
\begin{proof} 
We start by defining the class $\mathcal{H}$. First, we partition all even numbers in blocks of increasing sizes: $I_1=\{2\}$, $I_2=\{4,6\}$, $I_3=\{8,10,12\}$ and so on so that the size of $I_k=\{n_{k1},\ldots,n_{kj},\ldots n_{kk}$\} is $k$. Then, for every $k\ge 1$ and $1\le j\le k$ we let 
\[
h_{kj}(n) = \begin{cases} 1 & \text{ if } n \in I_k\setminus\{n_{kj}\}\\
                    0 & \text{ otherwise }
                    \end{cases}
\] 
 and put it into $\mathcal{H}$. Then we consider a  total injective computable function $f\colon\mathbb{N}\to\mathbb{N}$ such that $f(\mathbb{N})$ is undecidable. For example, one can take any enumerable but not recursive set $S\subseteq \mathbb{N}$ and let $f(a)$ be the $a$th natural number in a  computable enumeration of $S$ (without repetitions to ensure that $f$ is injective). Now, for every $a\in\mathbb{N}$ we define a hypothesis $h_a$  which is  equal to 1 on $I_{f(a)}\cup\{2a+1\}$ and nowhere else. We put all hypotheses $h_a$ into $\mathcal{H}$ as well. This finishes the description of $\mathcal{H}$, which consists only of functions with finite support. 
 
 We now show that $\mathcal{H}$ is decidable. Let  $h\colon\mathbb{N}\to\mathbb{N}$ be any function with finite support.
We first check whether the support of $\mathcal{H}$ intersects $I_k$ for exactly one $k$ (otherwise $h\notin\mathcal{H}$). If it does, there are two possibilities for $h$ to be in $\mathcal{H}$. The first possibility is when $\supp(h)$ is a subset of $I_k$ of size $k - 1$. In this case $h = h_{kj}$ for some $j$ and thus $h\in\mathcal{H}$. The only other possibility is when $\supp(h) = I_k \cup\{2a + 1\}$ for some $a$. Then $h\in \mathcal{H}$ if and only if $k = f(a)$, and we check this by  computing $f(a)$.

We now show that $\mathcal{H}$ is properly CPAC learnable. By Proposition \ref{cpac_char}, it is enough to show two things: that $\mathcal{H}$ has finite VC dimension and that $\mathcal{H}$ has a computable asymptotic ERM. We now show that the VC dimension of $\mathcal{H}$ is at most 3. For that, we take any 4 distinct natural numbers $x_1, x_2, x_3, x_4$ and show that not all $2^4$ Boolean functions on $S = \{x_1, x_2, x_3, x_4\}$ can be realized by hypotheses from $\mathcal{H}$. First, observe that the support of every hypothesis from $\mathcal{H}$ intersects exactly one block $I_k$. Hence, if $S$ intersects two different blocks, we cannot realize the all-ones function on $S$. Likewise, we cannot realize the all-ones function if $S$ has two distinct odd numbers (every hypothesis from $S$ has at most one odd number in its support). The only case left is when $S$ intersects exactly one block $I_k$ and, besides that, possibly has exactly one odd number. Then at least 3 elements of $S$ are from $I_k$. Observe that no function from $\mathcal{H}$ can be equal to 0 on two of these elements and to 1 on the third one.

We now construct a computable asymptotic ERM $A$ for $\mathcal{H}$. Assume that $A$ receives on input a sample $S = ((x_1, y_1), \ldots, (x_m, y_m))$
of size $m$. Then $A$ works as follows. First, it constructs a finite set of hypotheses $H_S\subseteq \mathcal{H}$. Then it goes through all hypotheses of $H_S$ and outputs one which minimizes $L_S(h)$ among them. We will argue that 
\begin{equation}
\label{approx_s}
\min_{h\in 
H_S} L_S(h) \le \min_{h\in\mathcal{H}} L_S(h) + \varepsilon_m, \qquad \varepsilon_m = \frac{1}{\min \big[f(\mathbb{N}) \setminus f(\{1, \ldots, m\})\big]}.
\end{equation}
As the sequence $\{\varepsilon_m)\}_{m = 1}^\infty$ converges to $0$ as $m\to\infty$, this will prove that $A$ is an asymptotic ERM for $\mathcal{H}$. We now explain how $A$ constructs $H_S$. Let $M = \max\{x_1,\ldots, x_m\}$. First, $A$ puts into $H_S$ one of the hypotheses $h\in H$ that equals $0$ on $\{1, \ldots, M\}$ (say $h = h_{M1}$ for instance). Then $A$ puts into $H_S$ all hypotheses of the form $h_{kj}$ such that $I_k$ intersects $\{1, \ldots, M\}$. Note that there are finitely many of them because $I_k$ is disjoint from $\{1, \ldots, M\}$ for $k\ge M$. Finally, $A$ puts $h_1, \ldots, h_{\max\{m, M\}}$ into $H_S$. To- compute these hypotheses, $A$ computes $f(1), \ldots,f(\max\{m, M\})$.

To show that $H_S$ satisfies \eqref{approx_s}, we take any $h\in H$ and show that there exists $h^\prime\in H_S$ with $L_S(h^\prime) \le L_S(h) + \varepsilon_m$. Assume first that $h = h_{kj}$ for some $k,j$. Note that there are only finitely many hypotheses of this form whose support intersects $\{1, \ldots, M\}$, and all of them are in $H_S$, so we can pick $h^\prime = h $ in this case. If it happens that $\supp(h_{kj})$ is disjoint from $\{1,\ldots,M\}$, then we can simply take $h^\prime=h_{M1}$ which coincides with $h_{kj}$ on $\{1,\ldots,M\}$.

Next, assume that $h = h_a$ for some $a$. If $a\le \max\{m, M\}$, then $h$ is also in $H_S$, so there is nothing left to do in this case. Likewise, we are done if $h_a$ equals $0$ on $\{1, \ldots, M\}$. The only remaining case is when $a> \max\{m, M\}$ and $\supp(h_a)$ intersects $\{1, \ldots, M\}$. Recall that $\supp(h_a) = I_{f(a)}\cup \{2a + 1\}$. Since $a > M$, we see that $I_{f(a)}$ must intersect $\{1, \ldots, M\}$ in this case. Hence, $H_S$ contains all hypotheses of the form $h_{f(a)j}$. Note that any of these $h_{f(a)j}$ differs from $h_a$ exactly at two points: $2a+1$ and the $j$th element of $I_{f(a)}$. Hence, the difference between $L_S(h_a)$ and $L_S(h_{f(a)j})$ can be bounded by the number of times these two points appear in the sample (divided by $m$, the size of the sample). Since $a > M$, the point $2a+1$ actually does not appear in the sample, and there exists $j$ such that the $j$th element of $I_{f(a)}$ appears in $S$ at most $m/|I_{f(a)}| = m/f(a)$ times. Hence, $L_S(h_a) \le L_S(h_{f(a)j}) + \frac{1}{f(a)}$ for some $j$. Finally, since $a > m$ and $f$ is injective, we have that $f(a) \in f(\mathbb{N}) \setminus f(\{1, \ldots, m\})$ and hence
\[
\frac{1}{f(a)} \le \frac{1}{\min \big[f(\mathbb{N}) \setminus f(\{1, \ldots, m\})\big]}  = \varepsilon_m.
\]

It remains to show that $\mathcal{H}$ is not properly SCPAC learnable. Equivalently, by Theorem \ref{thm:proper-SCPAC}, we have to show that $\mathcal{H}$ does not have a computable ERM. Assume for contradiction that it does, and call it $A$. Now let $k\in\mathbb{N}$ be any natural number and take 
 a sample $S$ of size $k$, which has each element of $I_k$ exactly once, all labeled by $1$. The only possible hypothesis in $\mathcal{H}$ which has $0$-error on $S$ (equivalently, is equal to 1 everywhere on $I_k$) is $h_a$ for $a$ with $f(a) = k$. Thus, such a hypothesis exists if and only if $k\in f(\mathbb{N})$, which can then be decided by running $A$ on $S$, and checking whether the output hypothesis has $0$-error on $S$, a contradiction. 
\end{proof}

\section{Proof of Theorem \ref{thm:finiteLdim}} \label{finiteLdim}
\label{sec:online}
\begin{proof}
We start by observing that if all the hypotheses in a class $\mathcal{H}$ have pairwise disjoint supports, then the Littlestone dimension of this class is at most 1. Indeed, consider any tree of depth 2. We show that $\mathcal{H}$ does not shatter it. Assume that its root is labeled by $x\in\mathbb{N}$. Note that there is at most one hypothesis $h\in\mathcal{H}$ with $h(x) = 1$. Hence, only one leaf under the right child of $x$ can be reached via a hypothesis from $\mathcal{H}$. This means that $\mathcal{H}$ does not shatter this tree.

By Proposition \ref{k-wit}, every CPAC learnable class admits a computable $k$-witness of VC dimension, for some $k\in\mathbb{N}$. We will construct a class that is not CPAC learnable by diagonalizing against all possible computable $k$-witnesses, for all $k$. Since we cannot computably  enumerate only potential witnesses, we diagonalize instead against all partial computable functions. The goal of our construction is to guarantee that for every potential computable $k$-witness $w$, there is a tuple of natural numbers and a hypothesis from $\mathcal{H}$ which agrees with $w$ on this tuple.

Let $((M_e, k_e))_{e \in \mathbb{N}}$ be a computable enumeration of all pairs of the form $(M, k)$, where $M$ is a Turing machine and $k$ is a natural number. Partition the set of even numbers into consecutive blocks $(I_e)_{e\in\mathbb{N}}$, where  the size of $I_e$ is $k_e$. Let $I_e(i)$ denote the $i$th smallest element of $I_e$. Also fix a computable bijection $p\colon\mathbb{N}^2\to\mathbb{N}$. 

We now define the class of hypotheses $\mathcal{H}$. For every $e, s\in\mathbb{N}$ such that $M_e$ halts on (a code for) the tuple $(I_e(1), \ldots, I_e(k_e))$ in exactly $s$ steps and outputs a binary  word $x = x_1\ldots x_{k_e}$ of length $k_e$, we define $h_{es}:\mathbb{N}\to\{0,1\}$ by
\[
h_{es}(n) = \begin{cases} 1 & \text{ if } n=2p(e,s)+1\\
x_i & \text{ if } n= I_e(i)\\   
0 & \text{ otherwise }
                    \end{cases}
\] 
and let $\mathcal{H}$ be the collection of all such hypothesis. Notice that every $h$ in $\mathcal{H}$ has finite support. We claim that, moreover, any two hypotheses from $\mathcal{H}$ have disjoint supports. Indeed, for any $e$, there is at most one $s$ such that $h_{es}$ is in our class. Hence, for any block $I_e$ of even numbers, there is at most one hypothesis whose support intersects this block. Thus, two distinct hypotheses cannot have a common even number in their support. Likewise, each odd number can be in the support of at most one $h_{es}$, because $p$ is a bijection. Hence, $\mathcal{H}$ has Littlestone dimension 1. Moreover, by similar considerations as in the proof of Theorem \ref{separation}, the decidability of whether a given finitely supported $h$ belongs to $\mathcal{H}$ boils down to a finite computation, in this case, the simulation of the first $s$ steps of computation of $M_e$ on $(I_e(1), \ldots, I_e(k_e))$.  

To complete the proof, suppose for contradiction that $\mathcal{H}$ is CPAC learnable and $w$ is a computable $k$-witness of  VC dimension for $\mathcal{H}$. There exists $e$ such that $M_e$ realizes $w$ and $k_e = k+1$. In particular, there exists $s$ such that $M_e$ halts on $(I_e(1), \ldots, I_e(k_e))$ in exactly $s$ steps and outputs $x = w((I_e(1), \ldots, I_e(k_e))$. Observe that $h_{es}$ belongs to $\mathcal{H}$ but agrees with $w$ on $(I_e(1), \ldots, I_e(k_e))$, a contradiction.

\end{proof}


\bibliographystyle{apalike}
\bibliography{main.bib}

\end{document}